\documentclass[11pt,letterpaper]{article}
\usepackage{amsfonts}
\usepackage{amsthm}
\usepackage{amsmath}
\usepackage[english]{babel}
\usepackage{url}
\usepackage[symbol]{footmisc}
\usepackage{mathtools}
\usepackage{algorithm}
\usepackage{algpseudocode}
\usepackage[letterpaper, margin=1in]{geometry}
\usepackage{mathtools}

\newtheorem{theorem}{Theorem}[section]

\newtheorem{lemma}[theorem]{Lemma}
\newtheorem{definition}[theorem]{Definition}

\DeclareMathOperator{\E}{\mathbb{E}}
\DeclareMathOperator{\Tr}{Tr}
\DeclareMathOperator{\Prob}{\mathbb{P}}
\DeclareMathOperator{\Rep}{Rep}

\DeclarePairedDelimiter{\ceil}{\lceil}{\rceil}

\title{	Fast computation of p-values for the permutation test based on Pearson's correlation coefficient and other statistical tests\footnote{This work was partly accomplished as a free project at Nebion AG.}}
\author{
Jean-Marie Droz\footnote{Nebion AG, Switzerland. \tt{droz.jm@gmail.com}.}
}

\date{\today}

\begin{document}
\maketitle

\begin{abstract} 

Permutation tests are among the simplest and most widely used statistical tools. Their p-values can be computed by a straightforward sampling of permutations. However, this way of computing p-values is often so slow that it is replaced by an approximation, which is accurate only for part of the interesting range of parameters. Moreover, the accuracy of the approximation can usually not be improved by increasing the computation time.

We introduce a new sampling-based algorithm which uses the fast Fourier transform to compute p-values for the permutation test based on Pearson's correlation coefficient. The algorithm is practically and asymptotically faster than straightforward sampling. For input size $n$ and target accuracy $\epsilon$, its complexity is $O(\frac{\log n}{\epsilon^2})$ against $O(\frac{n}{\epsilon^2})$ for the straightforward sampling approach. (We assume that $\epsilon$ is sufficiently small, as will be the case in most applications.) The idea behind the algorithm can also be used to accelerate the computation of p-values for many other common statistical tests. The algorithm is easy to implement, but its analysis involves results from the representation theory of the symmetric group.
\end{abstract}

\newpage
\section{Introduction}
For many commonly used statistical tests, notably permutation tests, p-values are too difficult to compute exactly, so that most statistical software use approximate calculations instead. Two main kinds of approximations can be used: 1) Approximations based on sampling\footnote{What we call sampling is sometimes called ``resampling''.} (the Monte-Carlo method) and 2) approximations based on replacing a probability distribution by a simpler one, often the normal distribution.
 The second kind of approximation typically yields fast algorithms, but is justified only asymptotically. The accuracy of methods using the second kind of approximation is also more difficult to evaluate and cannot be tuned. In contrast, approximation algorithms based on sampling are usually slow, but can be made as accurate as required at the price of an increase of their running time.  

The present article introduces a new method for the computation of p-values for numerous permutation-based statistical tests. Our method is based on approximation by sampling and conserves its advantages: the accuracy of the estimates can be tuned easily and the algorithms are simple to implement. However, our method yields algorithms that are much faster than straightforward sampling. We will prove our method's effectiveness in a simple case of practical importance: Pearson's correlation coefficient. 

The rest of the introduction will explain our main example. Section \ref{section:MainResult} introduces our method through its application to a statistical test based on Pearson's correlation. Sections \ref{section:covariancebound}, \ref{section:checkerboard}, \ref{section:dimensionbound}, and \ref{section:characterbound} establish the complexity of the algorithm for the p-value of the statistical test based on Pearson's correlation coefficient. Sections \ref{section:application} and  \ref{section:conservativeness} explain how p-value computations for other statistical tests can be accelerated and show how to make the p-values exact or ``conservative'' in a precise sense.

\subsection{The permutation test for correlation coefficients}

For $i\in \{1,\ldots,n\}$, let $U_i,V_i\in \mathbb{R}$ be pairs of samples of two possibly dependent random variables. Let $U,V\in \mathbb{R}^n$ be the two vectors representing the samples. 
The Pearson's correlation coefficient of $U$ and $V$ is given by $\frac{Cov(U,V)}{\sqrt{Var(U)Var(V)}}$. It constitutes a natural measure of the correlation between the two random variables. 
Pearson's correlation coefficient can be used to build a permutation test for the null-hypothesis that the two random variables are independent. 
The p-value of the test is given by $\Prob(\sigma U\cdot V\geq U \cdot V)$, where $\sigma$ is a permutation drawn from the uniform distribution over the symmetric group $\mathbb{S}_n$ and the permutations act on vectors by permuting their components.

The task of computing a p-value for the correlation coefficient is equivalent\footnote{The second task may seem more general, but it can be reduced to the first one by adding two components to the vectors $U$ and $V$.}
to the task of computing $\Prob(\sigma U\cdot V\geq t)$ for an arbitrary $t\in \mathbb{R}$.
We will mainly consider this more general task. Since we can reduce the problem of counting solutions to 0-1 knapsack with $n$ items to computing $n$ p-values of correlation coefficients, computing the p-value exactly is $\#P$-hard. However, the p-value can be approximated using a random algorithm. A simple Monte-Carlo sampling of the permutation $\sigma$ followed by a computation of correlation coefficients allows a randomized computation of the p-value in time $O(\frac{n}{\epsilon^2})$, where the algorithm is allowed to err by more than $\epsilon$ with probability inferior to $\frac{1}{3}$.

\section{Main result}
\label{section:MainResult}
We give an algorithm to compute the p-value of the correlation coefficient in time $O(max(n\log n,\frac{\log n}{\epsilon^2}))$. 

\begin{algorithm}
\caption{P-value estimation}
\label{alg:A}
\begin{algorithmic}

\Require{Two vectors $U,V\in \mathbb{R}^n$, an accuracy $\epsilon > 0$ and a probability of failure $\delta>0$}
\Ensure{An approximation $\bar{x}$ of $\Prob(\sigma U\cdot V\geq t)$ such that the error is greater than $\epsilon\cdot \sqrt{\Prob(\sigma U\cdot V\geq t)}$ with a probability bounded above by $\delta$}

\For{$i \in \{1,\dots, \ceil{\frac{C}{\delta n\epsilon^2}}\}$ for a constant $C$ to be chosen later}
        \State {Pick permutations $\sigma_1,\sigma_2$ at random. Compute $y_{i,k} = \sigma_1 U\cdot\lambda^k\sigma_2 V$ for $k\in \{0,\ldots,n-1\}$ by observing that the vector $y_{i,\cdot}$ is the product of the vector $\sigma_1 U$ and a circulant matrix representing $\lambda^k\sigma_2 V$ for all $k$'s. 
Multiplication by a circulant matrix can be done in time $O(n\log n)$ using the fast Fourier transform. 
Set $x_i$ to be the average of the numbers $z_{i,k} = 1_{\geq t} y_{i,k}$ over $k\in \{0,\ldots,n-1\}$}.
\EndFor\\
\Return The average $\bar{x}$ of the $x_i$'s
\end{algorithmic}
\end{algorithm}

The idea behind Algorithm \ref{alg:A} is simple: If the $z_{i,k}$ were independent for different $k$, the algorithm would have the same variance as the naive Monte-Carlo algorithm but would do $n$ sampling steps in time $O(n\log n)$. 
We thus need to prove that the $z_{i,k}$ for different $k$'s have sufficiently small covariance that an application of Chebyshev's inequality yields the correctness of the algorithm. 
More precisely, we need to show that the variance of $x_i$ is $O(n^{-1})$, as will be proved in Theorem \ref{th:covariance}.

\begin{theorem}
Let $p = \Prob(\sigma U\cdot V\geq t)$. Let $\bar{x}$ be computed as in Algorithm \ref{alg:A}. The probability that $|\bar{x}-p|>\epsilon\cdot \sqrt{p}$ is smaller than $\delta$. 
\end{theorem}
\begin{proof}
For $i \in \{1,\dots, \ceil{\frac{C}{\delta n\epsilon^2}}\}$, $x_i = \frac{\sum_{k=1}^n z_{i,k}}{n}$. 
So that $$Var(x_i) = \frac{\sum_{k=1}^n Var(z_{i,k}) + 2\sum_{1\leq j<k\leq n} Cov(z_{i,k},z_{i,j})}{n^2}.$$ 
By Theorem \ref{th:covariance}, $\frac{\sum_{1\leq j<k\leq n} Cov(z_{i,k},z_{i,j})}{n^2} = O(\frac{p}{n})$. 
Since the $z_{i,k}$'s have values in $\{0,1\}$ and expectation $p$, $Var(z_{i,k})=p(1-p)$. 
We deduce that $Var(x_i) = O(\frac{p}{n})$ and thus, since the $x_i$'s are independent, we can choose $C$ such that $Var(\bar{x})\leq C\frac{p}{n\ceil{\frac{C}{\delta n\epsilon^2}}}\leq\epsilon^2p\delta$. 
By Chebyshev's inequality, $\Prob(|\bar{x}-p|>t)\leq\frac{Var(\bar{x})}{t^2}$. Substituting, 
our upper bound for the variance and setting $t = \epsilon\cdot \sqrt{p}$, we get 
$\Prob(|\bar{x}-p|>t)\leq \delta$.
\end{proof}

The dependence of Algorithm \ref{alg:A} on the failure probability $\delta$ is linear: the running time is $O(\frac{n}{\delta\epsilon^2})$. 
It can be improved to $O(\log(\delta^{-1})\frac{n}{\epsilon^2})$ by running Algorithm \ref{alg:A}  $O(\log(\delta^{-1}))$ times with a small constant probability of failure and taking the median of the estimates. This is the classical ``median trick''.

\section{The covariance bound}
\label{section:covariancebound}
We denote by $[\ldots,a_i^{b_i},\ldots]$ the conjugacy class of permutations with $b_i$ cycles of length $a_i$.

\begin{theorem}
\label{th:covariance}
Let $p=\Prob(\sigma U\cdot V\geq t)$ for $\sigma$ a random permutation. 
Let $U,V$ be two vectors of length $n$ and $t\in \mathbb{R}$. 
Let $\sigma_1,\sigma_2$ be random permutations. 
We use $\lambda$ to denote a fixed permutation with one cycle of length $n$. 
Let $k$ be a uniform random variable over $\{1,\ldots,n-1\}$. 
If $\sigma_1,\sigma_2$ and $k$ are independent, then 
\begin{equation} \label{eq:main}
Cov(1_{\geq t}(\sigma_1 U\cdot\sigma_2 V),1_{\geq t}(\sigma_1 U\cdot\lambda^k\sigma_2 V)) = O(\frac{p}{n}).
\end{equation}
\end{theorem}

\subsection{Intuition behind the theorem and its proof}

We call {\it trivial} the largest eigenvalue of the adjacency matrix of a graph. It is always equal to the degree of the graph.
We rewrite $Cov(1_{\geq t}(\sigma_1 U\cdot\sigma_2 V),1_{\geq t}(\sigma_1 U\cdot\lambda^k\sigma_2 V))$ as $Cov(1_{\geq t}(\sigma_1 U\cdot V),1_{\geq t}(\sigma_1 U\cdot\sigma_2^{-1}\lambda^k\sigma_2 V))$. Let's fix $k$. The permutation $\sigma_2^{-1}\lambda^k\sigma_2$ is a random element of the conjugacy class of $\lambda^k$. Let $G$ be the Cayley graph on $\mathbb{S}_n$ with generators the conjugates of $\lambda^k$. Let $S$ be the subset of vertices $\sigma$ of $G$ for which $(\sigma U\cdot V)\geq t$. We could try to bound $Cov(1_{\geq t}(\sigma_1 U\cdot V),1_{\geq t}(\sigma_1 U\cdot\sigma_2^{-1}\lambda^k\sigma_2 V))$ by establishing that $G$ is edge expanding for $S$ and its translates. This would follow from the expander mixing lemma, if all non-trivial eigenvalues of $G$ were in an appropriate interval. This almost works: all eigenvalues of (the adjacency matrix of) $G$ except $-|[\lambda^k]|$\footnote{Minus the size of the conjugacy class of the k-th power of $\lambda$} are in the correct interval, for most values of $k$.

We show that the bad values of $k$ are sufficiently few and that the set $S$ is sufficiently uncorrelated with the eigenvector corresponding with the inconvenient eigenvalue $-|[\lambda^k]|$ that an argument inspired by the proof of the expander mixing lemma is enough to deduce our theorem. It is known since \cite{Lovasz75} that the eigenvalues of a Cayley graph are related to representations of its group and it will thus not be a surprise that bounds on characters play an essential role in the proof. The correspondence between eigenvalues of the graph and representations is especially simple when, as is the case here, the generating set of the Cayley graph is closed by conjugation \cite{Diaconis81}.

 The rest of the present section contains the core of the proof. Section \ref{section:checkerboard} contains a combinatorial argument which allows us to deal with the problematic eigenvalue $-|[\lambda^k]|$ or, equivalently, the alternating representation. Sections \ref{section:dimensionbound} and \ref{section:characterbound} prove bounds, on dimensions and characters of representations of the symmetric group, that are used in the proof.

\begin{proof}
We can assume $k\neq \frac{n}{2},\frac{n}{3},\frac{2n}{3}$. 
Indeed, the contributions of those at most three possible $k$'s is bounded by $Cov_{k\in\{\frac{n}{2},\frac{n}{3},\frac{2n}{3}\}}(1_{\geq t}(\sigma_1 U\cdot\sigma_2 V),1_{\geq t}(\sigma_1 U\cdot\lambda^k\sigma_2 V))\Prob(k\in\{\frac{n}{2},\frac{n}{3},\frac{2n}{3}\})\leq\frac{3p}{n}$. The inequality holds because the covariance is bounded by $p$. 

We observe that $Cov(1_{\geq t}(\sigma_1 U\cdot\sigma_2 V),1_{\geq t}(\sigma_1 U\cdot\lambda^k\sigma_2 V))$ can be rewritten $$Cov(1_{\geq t}(\sigma_2^{-1}\sigma_1 U\cdot V),1_{\geq t}(\sigma_2^{-1}\sigma_1 U\cdot\sigma_2^{-1}\lambda^k\sigma_2 V))$$ and then $$Cov(1_{\geq t}(\sigma_1 U\cdot V),1_{\geq t}(\sigma_1 U\cdot\sigma_2^{-1}\lambda^k\sigma_2 V)).$$ This follows from the fact that for independent and uniformly distributed random variables $\sigma,\sigma'$ over $\mathbb{S}_n$, the pairs $(\sigma,\sigma')$ and $(\sigma,\sigma\sigma')$ obey the same distribution.
We use the representation theory of the symmetric group to bound the left hand side of Equation (\ref{eq:main}). 
If we write $X_0=1_{\geq t}(\sigma_1 U\cdot V)$, $X'_k=1_{\geq t}(\sigma_1 U\cdot\sigma_2^{-1}\lambda^k\sigma_2 V)$ and $X_k=\E_{\sigma_2}X'_k$, we have:
\begin{equation} \label{eq:cov}
Cov(X_0,X'_k) = \E X_0 X'_k -\E X_0 \E X'_k 
\end{equation}
\begin{equation} \label{eq:cov2}
=\frac{1}{n!}\sum_{\sigma\in\mathbb{S}_n} X_0(\sigma)X_k(\sigma)-p^2 = \frac{1}{n!^2}\sum_{\rho\in\Rep(\mathbb{S}_n)} d_\rho \Tr(\hat{X_0}^*(\rho)\hat{X_k}(\rho)) - p^2
\end{equation}
where $\hat{X_0}$ and $\hat{X_k}$ are the Fourier transform\footnote{The Fourier transform of $f:\mathbb{S}_n\rightarrow \mathbb{R}$ is $\hat{f}(\rho)=\sum_{\sigma\in\mathbb{S}_n}f(\sigma)\rho(\sigma)$ where $\rho$ is a representation of $\mathbb{S}_n$.} of $X_0$ and $X_k$ regarded as functions over the permutations: $X_0(\sigma) = 1_{\geq t}(\sigma U\cdot V)$ and $X_k = X_0 * \frac{1_{\in [a^{\frac{n}{a}}]}}{|[a^{\frac{n}{a}}]|}$  with $a=\frac{n}{gcd(n,k)}$. The symbol $*$ denotes the convolution and $1_{\in [a^{\frac{n}{a}}]}$ represents the indicator function on the set of conjugates $[a^{\frac{n}{a}}]$.
The last equality is the Plancherel formula for the discrete Fourier transform on finite groups. 
Note that the character of a permutation and of its inverse are equal, because they are conjugate. This allows small simplifications in some of our formulas.
Since class functions transform to multiple of the identity, $$\frac{\widehat{1_{\in [a^{\frac{n}{a}}]}}(\rho)}{|[a^{\frac{n}{a}}]|} = \frac{\chi_{\rho}([a^{\frac{n}{a}}])}{d_\rho}I,$$ where $\chi_\rho$ is the character of the representation $\rho$ and $I$ is the identity matrix.

Because convolution corresponds to matrix multiplication of the Fourier transforms, we have: $$\hat{X_k}(\rho)=\frac{\hat{X_0}(\rho)\widehat{1_{\in [a^{\frac{n}{a}}]}}(\rho)}{|[a^{\frac{n}{a}}]|}=\frac{\hat{X_0}(\rho)\chi_{\rho}([a^{\frac{n}{a}}])}{d_\rho}.$$
 
We can then write:
\begin{equation} 
Cov(X_0,X'_k) = \sum_{\rho} \Tr(\hat{X_0}^*(\rho)\hat{X_0}(\rho))\frac{\chi_{\rho}([a^{\frac{n}{a}}])}{n!^2} - p^2
\end{equation}
For the trivial representation, we have $$\Tr(\hat{X_0}^*(1)\hat{X_0}(1))\frac{\chi_{1}([a^{\frac{n}{a}}])}{n!^2}=p^2,$$ because $\hat{X_0}(1) = pn!\cdot I$ and $\chi_1(\sigma)=1$ for any permutation $\sigma$. 

We also have $p\geq Cov(X_0,X_0)=\sum_{\rho\neq 1} d_\rho \Tr(\hat{X_0}^*(\rho)\hat{X_0}(\rho))\frac{1}{n!^2}$. 
This expresses $p$ as a sum of non-negative terms. Denoting by $\Lambda$ the alternating representation, we deduce that 
\begin{align}
Cov(X_0,X'_k) &= \Tr(\hat{X_0}^*(\Lambda)\hat{X_0}(\Lambda))\frac{1}{n!^2} + \sum_{\rho\notin 1,\Lambda} \Tr(\hat{X_0}^*(\rho)\hat{X_0}(\rho))\frac{\chi_{\rho}([a^{\frac{n}{a}}])}{n!^2}\\
 &\leq \Tr(\hat{X_0}^*(\Lambda)\hat{X_0}(\Lambda))\frac{1}{n!^2} + \max_{\rho\notin 1,\Lambda} \frac{|\chi_{\rho}([a^{\frac{n}{a}}])|}{d_\rho}\sum_{\rho\notin 1,\Lambda} d_\rho \Tr(\hat{X_0}^*(\rho)\hat{X_0}(\rho))\frac{1}{n!^2} \\
&\leq \Tr(\hat{X_0}^*(\Lambda)\hat{X_0}(\Lambda))\frac{1}{n!^2}+
p \max_{\rho\notin 1,\Lambda} \frac{|\chi_{\rho}([a^{\frac{n}{a}}])|}{d_\rho}.\label{eq:pmax}
\end{align}
By Theorem \ref{thm:charachterbound}, when $a\geq 4$, 
\begin{equation}
\label{eq:charachterbound}
\max_{\rho\notin 1,\Lambda} \frac{|\chi_{\rho}([a^{\frac{n}{a}}])|}{d_\rho}\leq \frac{3}{n}
\end{equation}
Therefore, the second summand of (\ref{eq:pmax}) is $O(\frac{p}{n})$.

It remains to bound the term coming from the alternating representation. Since the alternating representations takes values in $\{-1,1\}$,

\begin{equation}
|\Tr(\hat{X_0}^*(\Lambda)\hat{X_k}(\Lambda))\frac{1}{n!^2}|=  |\Tr(\hat{X_0}^*(\Lambda)\hat{X_0}(\Lambda))\frac{\chi_{\Lambda}([a^{\frac{n}{a}}])}{n!^2}| = |\Tr(\hat{X_0}^*(\Lambda)\hat{X_0}(\Lambda))\frac{1}{n!^2}| 
\end{equation}

On the one hand, $|\hat{X_0}(\Lambda)|\leq |X_0|=n!p$.
On the other hand, assuming without loss of generality that $U$ and $V$ have components in increasing order, we can use Lemma \ref{lem:isdownset} and \ref{lem:towers2} to deduce $|\hat{X_0}(\Lambda)|\leq \frac{n!}{n}$. Therefore,

\begin{equation}
\label{eq:checkerboardbound}
\Tr(\hat{X_0}^*(\Lambda)\hat{X_k}(\Lambda))\frac{1}{n!^2}\leq \frac{p}{n}.
\end{equation}

Taken together, the estimates  (\ref{eq:charachterbound}), (\ref{eq:checkerboardbound}) and the chain of inequalities ending in Equation (\ref{eq:pmax}) imply our theorem.
\end{proof}

The actual constant hidden in the big-O notation in the statement of the theorem can be shown to be small. It could even be $0$, we did not find any pair $U,V$ for which the covariance was positive.

\section{Upper sets on checkerboards}
\label{section:checkerboard}
\begin{definition} 
The {\it discrepancy} of a subset $S$ of $\mathbb{Z}^n$ is the difference between the number of points in $S$ with even sum of coordinates and the number of points in $S$ with odd sum of coordinates.
\end{definition}
\begin{definition} 
We write a permutation $\sigma$ of $n$ as $(a_1,\ldots ,a_n)$, meaning that $\sigma(a_i)=i$. The {\it factorial lattice} is a partial order on the permutations of size $n$. It is the transitive closure of the relation $\prec$, where $a\prec b$ for $a=(a_1,\ldots ,a_n)$ and $b=(b_1,\ldots ,b_n)$ if and only if there exist $i$ and $j$ with $i<j$ such that $a_i>a_j$, $a_i=b_j$, $a_j=b_i$, for $ k \notin \{i,j\}$ $a_k=b_k$, and for $i<k<j$, $a_k>a_i$.
\end{definition}

\begin{lemma}
\label{lem:factorial}
The factorial lattice is the product of the totally ordered sets on $k$ elements $I_k$ for $k\in\{1\ldots n\}$. 
\end{lemma}
\begin{proof}
Associate to a permutation $(a_1,\ldots ,a_n)$ the tuple $(l_1,\ldots ,l_n)$ where $l_i$ is the number of $a_j<a_i$ with $j<i$.
\end{proof}

From this last proof, $\mathbb{S}_n$ can be viewed as embedded in $\mathbb{Z}^n$. The induced notion of discrepancy on sets of permutations does not depend on the embedding.

\begin{lemma}
\label{lem:isdownset}
For vectors with components in increasing order $U,V\in \mathbb{R}^n$ and $t\in \mathbb{R}$, the set $S=\{\sigma|\sigma U\cdot V\geq t,\sigma \in \mathbb{S}_n\}$ is an upper set in the factorial lattice.
\end{lemma}
\begin{proof}
The transpositions generating the factorial lattice cannot decrease the scalar product: $\sigma\prec \sigma'$ implies $\sigma U\cdot V\leq \sigma' U\cdot V$. 
\end{proof}

\begin{lemma}
\label{lem:towers1.5}
The discrepancy of an upper set is at most $\frac{n!}{n}$.
\end{lemma}
\begin{proof}
The projection of an upper set on the space perpendicular to the $n$-th coordinate along the coordinate (the direction in which the factorial lattice is the ``longest'') has size at most $\frac{n!}{n}$. The discrepancy of a ``column'' of points projecting to the same place is at most $1$.
\end{proof}

\begin{lemma}
\label{lem:towers2}
The Fourier transform $\widehat{1_{\in S}}(\Lambda)$ of the indicator function of an upper set $S$ evaluated at the alternating representation is at most $\frac{n!}{n}$.
\end{lemma}
\begin{proof}
For $\sigma\in\mathbb{S}_n$, $\Lambda(\sigma)$ has absolute value $1$ and its sign corresponds to the parity of the number of inversions in $\sigma$. The embedding of $\mathbb{S}_n$ in $\mathbb{Z}^n$ of the proof of Lemma \ref{lem:factorial} implies that the discrepancy of $S$ is, up to a sign, equal to $\sum_{\sigma \in S} \Lambda(\sigma)=\widehat{1_{\in S}}(\Lambda)$. We finally apply the previous lemma.
\end{proof}

\section{Dimension estimates of representations}
\label{section:dimensionbound}

The representations of the symmetric group $\mathbb{S}_n$ are in bijection with the partitions of $n$, or equivalently, the Young diagrams with $n$ boxes. The dimension of a representation equals the number of ways to fill the associated young diagram to obtain a standard Young tableau. From the definition of standard Young tableaux, intuitively, for fixed $n$, representations of $\mathbb{S}_n$ associated with Young diagrams in which the boxes are either almost all in one column or almost all in one row should have the smallest dimension, because there is less freedom in how to construct a Young tableau inside. The next theorem makes this intuition quantitative. It is easy and probably well known in some circles, but we could not locate an appropriate reference. The first three chapters of \cite{Sagan} cover all the concepts and results related to the representation theory of the symmetric group that we use in this section and much more. 

\begin{theorem}
\label{thm:dimbound}
For $n\geq 400$, every representation $\rho$ of $\mathbb{S}_n$ has dimension $d_\rho>\frac{n^2}{3}$ except $1,\Lambda,(n-1,1),(n-1,1)^T$.
\end{theorem}
\begin{proof}
Let $t_\rho$ be the Young diagram associated to $\rho$. For the rest of the proof, we assume the leg\footnote{The first column of the Young diagram minus the upper left cell in English notation.} of length $l$ of $t_\rho$ is smaller or equal than the arm\footnote{The first row of the Young diagram minus the upper left cell in English notation.} of length $a$.
If $t_\rho$ does not have a box at coordinates $(2,2)$, a simple application of the hook-length formula shows $d_\rho=\binom{n-1}{l}$. 
Since $\rho \not\in \{1,\Lambda,(n-1,1),(n-1,1)^T\}$, $l\geq 2$.
In this case, our theorem follows from the unimodality\footnote{The only mode of a row of Pascal's triangle is in the middle.} of the binomial coefficients: $\binom{n-1}{l}\geq \binom{n-1}{2} \geq \frac{n^2}{3}$ for $n\geq 9$.

If $t_\rho$ does have a box at coordinates $(2,2)$, we count the number of standard Young tableaux in the arm, the leg and the $(2,2)$ box of $t_\rho$ only. This provides a lower bound on the number of standard Young tableaux in $t_\rho$ and hence a lower bound on $d_\rho$.
The two inequalities $(a+1)(l+1)\geq n$ and $d_\rho\geq \binom{a+l}{l+1}$ hold. If $l>5$, $n<(a+1)^2$ and thus $d_\rho \geq\binom{a+l}{l+1}\geq\binom{a}{6}\geq \frac{(\sqrt{n}-1)^6}{6!}$ by unimodality of the binomial coefficients. When $n\geq 400$, $\frac{(\sqrt{n}-1)^6}{6!}\geq \frac{n^2}{3}$ and thus $d_\rho\geq \frac{n^2}{3}$. If $2\leq l \leq 5$, $a\geq\frac{n}{5}$ and therefore $d_\rho \geq\binom{a+l}{l+1}\geq\binom{a+l}{3}\geq cn^3$ for some constant $c>0$ and our bound for $d_\rho$ is established for $n$ large enough. Taking $n\geq 400$ is enough to imply $d_\rho>\frac{n^2}{3}$, but we skip this tedious computation.

For $l=1$, let $a_2$ be the number of cells directly below the arm. We have $n=a+a_2+2$. By counting standard tableaux with numbers between $\{2,\cdots,a_2+1\}$ in the cells of the arm that have a cell below themselves, we deduce $d_\rho\geq \binom{a+1}{a_2+1}$. When $a_2<\frac{a}{2}$ and $n\geq 400$ is fixed, this last lower bound on $d_\rho$ is minimal for $a_2=1$. Therefore we can take $a_2=1$ and, assuming $a\geq 19$ insures that $d_\rho\geq \frac{n^2}{3}$. When $n\geq 40$, $l=1$ implies $a\geq 19$.
Finally, if $a_2\geq\frac{a}{2}$, we have $d_\rho\geq \binom{a}{\ceil*{\frac{a}{2}}}\geq \frac{2^a}{a+1}$  (the first inequality is again established by counting a subset of standard tableaux) and $n\leq 2a+2$ so that $n\geq 42$ implies $a\geq 20$, which implies $d_\rho \geq \frac{n^2}{3}$. 
\end{proof}

\section{Bounds on Characters}
\label{section:characterbound}
\begin{lemma}
\label{lem:n-1}
For $\rho \in \{(n-1,1),(n-1,1)^T\}$ a representation of $\mathbb{S}_n$ and a conjugacy class $[r^m]$ of $m$ identical cycles of size $r$ larger than $1$, $\frac{|\chi_\rho([r^m])|}{d_\rho}\leq \frac{1}{n-1}$. 
\end{lemma}
\begin{proof}
Applying the Murnaghan-Nakayama rule to the Young diagrams $(n-1,1)$ and $(n-1,1)^T$, we see that $|\chi_\rho([r^m])|=1$. Computing $d_\rho = n-1$ using the hook formula yields the claim.
\end{proof}

\begin{theorem}
For $n\geq 400$, $\rho \notin \{1,\Lambda\}$ and $n=rm$ with $r>1$, we have 
\begin{equation}
\frac{|\chi_{\rho}([r^{m}])|}{d_\rho}\leq \begin{cases}
 \frac{3}{n} & \textrm{when}\: r\geq 4 \\
 \frac{3}{n^\frac{1}{2}} & \textrm{when}\: r=2,3 
\end{cases}
\end{equation}
\label{thm:charachterbound}
\end{theorem}
\begin{proof}
For $\rho \in \{(n-1,1),(n-1,1)^T\}$, the theorem follows from Lemma \ref{lem:n-1} and $n\geq 400$. We thus assume that $\rho$ satisfies the condition of Theorem \ref{thm:dimbound}. 
From \cite{FominLulov95}, we have $$|\chi_\rho([r^m])| \leq \frac{m!r^m}{(mr)!^{\frac{1}{r}}}d_\rho^\frac{1}{r}.$$ 
We state two inequalities from \cite{Robbins55}, which refine Stirling's formula:
$$\sqrt{2\pi}n^{n+\frac{1}{2}}e^{-n+\frac{1}{12n+1}}<n!<\sqrt{2\pi}n^{n+\frac{1}{2}}e^{-n+\frac{1}{12n}}.$$
We deduce:
$$|\chi_\rho([r^m])| \leq e^\frac{1}{12}\frac{ \sqrt{2\pi}m^{m+\frac{1}{2}}e^{-m} r^m}{ (\sqrt{2\pi}(rm)^{rm+\frac{1}{2}}e^{-rm})^{\frac{1}{r}} } d_\rho^\frac{1}{r} $$
$$ \leq e^\frac{1}{12}\sqrt{2\pi} m^{\frac{1}{2}-\frac{1}{2r}}r^{-\frac{1}{2r}} d_\rho^\frac{1}{r} $$
$$ \leq 3m^{\frac{1}{2}-\frac{1}{2r}}d_\rho^\frac{1}{r}.$$
For $n=rm$, assuming $r=2,3$, we have $$|\chi_\rho([r^m])| \leq \frac{3}{2}n^{\frac{1}{2}}d_\rho^\frac{1}{2}.$$ 
Dividing by $d_\rho$ and using Theorem \ref{thm:dimbound} to replace $d_\rho$ by $\frac{n^2}{3}$ on the right hand side, we have: 
$$\frac{|\chi_\rho([r^m])|}{d_\rho} \leq 3n^{-\frac{1}{2}} .$$
Assuming $r>3$, we have $$|\chi_\rho([r^m])| \leq n^{\frac{1}{2}}d_\rho^\frac{1}{r}.$$
Dividing by $d_\rho$, we get $$\frac{|\chi_\rho([r^m])|}{d_\rho} \leq n^{\frac{1}{2}}d_\rho^\frac{1-r}{r}.$$ 
By Theorem \ref{thm:dimbound}, $d_\rho\geq \frac{n^2}{3}$ and thus, $$\frac{|\chi_\rho([r^m])|}{d_\rho} \leq 3n^{-1}$$
\end{proof}

\section{Applications}
\label{section:application}
Our main algorithmic idea can be used to compute the p-value for any statistic of the form $f(u_1\cdot \pi(v_1),\ldots, u_k\cdot \pi(v_k))$ for a random permutation $\pi$, a $k$-ary function $f$, and tuples $v$ and $u$ of length $k$. Its domain of application can be extended further with various problem-specific ideas.

The method is, for example, a promising candidate for computing p-values for the following tests. We did not try to be exhaustive and we only sketch how our method could be applied.

\begin{enumerate}
\item Spearman's and Pearson's correlation significance test, as shown above.
\item The Mann-Whitney U test. Given two samples $X,Y\subset \mathbb{R}$, the p-value for the Mann-Whitney U test can be calculated as $\Prob(\sigma U\cdot V\geq U\cdot V)$ for $V$ a vector containing the ranks of the elements of $X$ in $X\cup Y$ followed by the ranks of the elements of $Y$, $U$ a vector containing $|X|$ zeros followed by $|Y|$ ones\footnote{This natural permutation based formula is incorrect in the presence of ties. (Working with continuous distributions is usually considered a necessary assumption for applying the Mann-Whitney U test.)
This can be seen by taking the example of two groups both with two samples taken from independent and balanced Bernoulli random variables. 
The probability of both samples of one group being $1$ and both samples of the other group $0$ is $1/16$, while there are $24$ permutations, which precludes getting a p-value with a denominator multiple of $16$.}.
Our improved sampling algorithm is expected to become competitive with exact algorithms, including \cite{NagarajanKeich09}, when the sizes of the groups are sufficiently large.
\item The Kruskal-Wallis one-way analysis of variance for a small number $k$ of groups: This test generalizes the previous example and our method can be applied in a similar way. Samples of the test statistics are obtained using a vector $R$ containing the ranks of the observations of all the groups and 0-1 vectors $X_k$ representing the groups. The sampled test statistic can be computed from $\sigma R\cdot X_i$ for $i\in\{1,\ldots,k\}$ for the random permutation $\sigma$. 
\item The Wilcoxon signed-rank test (possibly with ties). The problem reduces to evaluating $\Prob(\Sigma_{i\in \{1,\ldots,n\}}i\cdot b_i\geq t)$ for balanced independent Bernoulli random variables $b_i$ and a threshold $t\in \mathbb{R}$. This evaluation can be approximated closely by computing $\Prob(\sigma U\cdot V\geq t)$ for $U$ a balanced 0-1 vector of length $c\cdot n$ for a small constant $c$ and $V$ a vector containing the ranks or equivalently, in the absence of ties, the integers $\{1,\ldots,n\}$ and a padding by zeroes. The approximation can be refined by using randomly slightly unbalanced 0-1 vectors for $U$. This approach probably gives a more efficient reduction of the problem of counting solutions to the 0-1 knapsack problem to the computing of correlation coefficient p-values.

\item A permutation test for distance correlation\footnote{See \cite{Szekely09} or \texttt{https://en.wikipedia.org/wiki/Distance\textunderscore correlation} for more about this beautiful new measure of dependence.} can also be defined and our method could be used to compute its significance. The basic idea is the following: If, for $i\in \{1,\ldots,n\}$, $U_i\in \mathbb{A},V_i\in \mathbb{B}$ are pairs of samples of two random variables $U$ and $V$ from metric spaces $(A,d_A)$, and $(B,d_B)$, under some conditions on the metrics, the correlation distance between the two random variables is, up to normalizations, given by $\Tr(M_U^TM_V)$ for $M_U,M_V\in \mathbb{R}^{n\times n}$, ``centered'' distance matrices of the samples. P-values for a test for the independence of $U$ and $V$ can then be obtained by comparing the correlation distance with the distribution of $\Tr(M_U^*\sigma(M_V))\cdot c$ for a normalization $c$ and $\sigma$ a random permutation acting on matrices by permuting their rows and their columns. We can then express $\Tr(M_U^*\sigma(M_V))\cdot c$ as a sum of the form $c\cdot \Sigma_{i\in \{1,\ldots,n\}} x_i\cdot \sigma(y_i)$ where $x_i$ and $y_i$ are vectors containing elements of the matrices $M_U$ and $M_V$. Our trick based on the the fast Fourier transform can then be applied.
\end{enumerate}

 In this article, we proved that our method provides a speed up for permutation tests based on Pearson's (and hence also Spearman's) correlation. Our results extend without difficulties to the Mann-Whitney U test and the Kruskal-Wallis test. For more complicated cases, it is likely that an important speed up occurs in practice on large problems, but the proof of a good complexity bound might require more work and could be the object of further publications.

\section{Conservative p-values with sampling tests}
\label{section:conservativeness}
An approximation to a p-value is usually not a p-value, see \cite{North09} for details. 
Algorithm \ref{alg:A}, which can be used to compute an approximate p-value, is modified to return proper (but non-deterministic) p-values. The result is Algorithm \ref{alg:B}. The p-values returned by Algorithm \ref{alg:B} are also approximations to the exact p-value of the permutation test and have therefore similar power. Naturally, to the extent that we can speak of the power of an approximation of a p-value, the result of Algorithm \ref{alg:A} also have similar power to the exact p-value.

\begin{algorithm}
\caption{Proper p-value}
\label{alg:B}
\begin{algorithmic}

\Require{Two vectors $U,V\in \mathbb{R}^n$, a number of iterations $i_{max}$}
\Ensure{A nondeterministic p-value for the null hypothesis $H_0$ that $U$ and $V$ are two sets of independent and identically distributed samples from two $1$-dimensional distributions $\mathbb{U}$ and $\mathbb{V}$}
\State{Set $t=U\cdot V$}
\State{Pick a random conjugate $\alpha$ of the long cycle $\lambda$. Set $x_0=\frac{\#\{k|U\cdot \alpha^k V\geq t  \}}{n}$.}
\For{$i \in \{1,\dots, i_{max}\}$}
        \State {Pick permutations $\sigma_1,\sigma_2$ at random. Compute $y_{i,k} = \sigma_1 U\cdot\lambda^k\sigma_2 V$ for $k\in \{0,\ldots,n-1\}$ by observing that the vector $y_{i,\cdot}$ is the product of the vector $\sigma_1 U$ and a circulant matrix representing $\lambda^k\sigma_2 V$ for all $k$'s. 
Set $x_i$ to be the average of the numbers $z_{i,k} = 1_{\geq t} y_{i,k}$ over $k\in \{0,\ldots,n-1\}$}.
\EndFor\\
\Return The average $\bar{x}$ of the $x_i$'s for $i \in \{0,\dots, i_{max}\}$
\end{algorithmic}
\end{algorithm}

Note that, as in Algorithm \ref{alg:A}, the computations of the $x_i$'s for $i \in \{0,\dots, i_{max}\}$ can be accelerated by using the fast Fourier transform. Therefore Algorithm \ref{alg:B} is as efficient as Algorithm \ref{alg:A}.

\begin{theorem}
Assuming the null hypothesis $H_0$ from Algorithm \ref{alg:B}, the algorithm returns each of the values $$\Big\{\frac{1}{n(i_{max}+1)},\ldots,\frac{n(i_{max}+1)}{n(i_{max}+1)}\Big\}$$ with the same probability. In particular, for any $\alpha\in[0,1]$, the probability under $H_0$ that the return value of Algorithm \ref{alg:B} is smaller or equal to $\alpha$ is smaller or equal to $\alpha$. In other words, the return value is a conservative p-value.
\end{theorem}
\begin{proof}
Under $H_0$, Algorithm \ref{alg:B} returns $\frac{r}{n(i_{max}+1)}$ for $r$ the rank of a random number among $n(i_{max}+1)$ random numbers sampled according to the same procedure.
\end{proof}

\section{Conclusions}\label{conclusions}

We have shown how the computation of the p-value for a statistical test can be made faster using the fast Fourier transform. The method can be applied to many statistical tests based directly or indirectly on the sampling of permutations. However, except for the correlation coefficient, further work is needed to know when the method really improves performances, in practice or asymptotically. Proving a good bound on the time complexity will in some case require new ideas. We leave open the problem of showing that the covariance in Theorem \ref{th:covariance} is actually always negative or finding a counterexample. This open problem hints at the existence of a different, maybe simpler proof of Theorem \ref{th:covariance}.

\section*{Acknowledgments}

I would like to thank Tomas Hruz from the Department of Computer Science of ETH Z\"urich for his proofreading, many great advices and his efficacious encouragements. Part of the research described in the present article was done at Nebion AG, largely as a free project. It would not have been possible without the enthusiasm and the friendly support of Stefan Bleuler the CTO at Nebion. Peter Widmayer, Gaston Gonnet and Peter B\"uhlmann from the Department of Computer Science and the Department of Statistics of ETH Z\"urich also have my gratitude for helpful discussions and advice.

\end{document}